\documentclass{amsart}

\usepackage[utf8]{inputenc}
\usepackage{amsmath,amsfonts,amssymb,mathrsfs,amstext,amscd,latexsym,amsthm, mathtools}
\usepackage{hyperref}
\usepackage{xcolor}
\usepackage{pgfplots}
\usepgfplotslibrary{fillbetween}
\usepackage{graphicx}
\usepackage{subcaption}
\usepackage{algorithm}
\usepackage{algpseudocode}

\usepackage{enumerate}

\newcommand{\G}{\mathcal{G}}
\newcommand{\E}{\mathcal{E}}
\newcommand{\V}{\mathcal{V}}

\newtheorem{theorem}{Theorem}[section]

\newtheorem{corollary}[theorem]{Corollary}
\newtheorem{proposition}[theorem]{Proposition}

\newtheorem*{remark}{Remark}
\newtheorem*{example}{Example}

\title[Ricci Curvature and Community Structure]{On the Relation between Graph Ricci Curvature and Community Structure}
\author{Theodora Bourni, Vasileios Maroulas, Sathyanarayanan Rengaswami}
\date{June 2024}

\usepackage{graphicx}

\pgfplotsset{compat=1.17}
\begin{document}

\begin{abstract}
The connection between curvature and topology is a very well-studied theme in the subject of differential geometry. By suitably defining curvature on networks, the study of this theme has been extended into the domain of network analysis as well. In particular, this has led to curvature-based community detection algorithms. In this paper, we reveal the relation between community structure of a network and the curvature of its edges. In particular, we give apriori bounds on the curvature of intercommunity edges of a graph.
\end{abstract}


\maketitle

\setcounter{tocdepth}{1}
\tableofcontents

\section{Introduction}

The connection between curvature and topology is a fundamental question in Riemannian geometry. Notable examples include the Gauss-Bonnet theorem, which relates the curvature of a surface to its Euler characteristic, and Myer's theorem, which bounds the diameter of a manifold in terms of its curvature (see \cite{lee2018introduction} for more details). Remarkably, successful definitions of curvature have been extended to graphs (Forman \cite{forman2003bochner}, Ollivier \cite{ollivier2009ricci}, Lin \cite{lin2011ricci}, Devriendt \cite{devriendt2022discrete}), generalizing the curvature notions on Riemannian manifolds and establishing analogous connections between curvature and topology.

A crucial topological question in the study of complex networks is their community structure (See \cite{fortunato2010community}, \cite{fortunato2016community} for a survey on community detection). This involves clustering nodes such that many edges connect nodes within the same cluster, while few edges connect nodes between different clusters. These clusters, referred to as communities, are defined based on the application at hand. This task is significant across various fields including computer science \cite{grout2006constrained, krishna1997cluster}, biology \cite{rives2003modular, spirin2003protein,, dunn2005use}, chemistry\cite{queen2023polymer}, logistics \cite{guimera2005worldwide, o1992clustering} , where graphs are commonly used to model real-world systems. Consequently, numerous methods based on diverse theories are available, such as partitioning algorithms and spectral methods (\cite{fortunato2010community} contains a survey of commonly used algorithms.)

A recent approach to community detection draws inspiration from the geometric notion of curvature. In this paper, we employ the Ollivier-Ricci curvature (ORC), originally defined by Ollivier \cite{ollivier2009ricci} using optimal transport theory. The essential idea is to compare the distance \(d(x,y)\) between two vertices \(x\) and \(y\) to the distance between the neighbors of \(x\) and \(y\) (defined in terms of optimal transport). If the latter distance is smaller, the edge is positively curved; if greater, it is negatively curved. Negatively curved edges act as ``bottlenecks," indicating that to move from the neighbors of \(x\) to those of \(y\), one must pass through the edge \(xy\). This concept is applied to rewire graph neural networks in \cite{topping2021understanding}.

Studies such as \cite{ni2019community} and \cite{sia2019ollivier} have observed that edges with positive ORC typically exist between nodes within the same community, whereas edges with negative ORC often connect nodes from different communities. This observation has been turned into a computational algorithm for community detection. \cite{sia2019ollivier} achieves this by sequentially deleting negatively curved edges, while \cite{ni2019community} employs ``Ricci flow with surgery" to elongate or contract edges based on their curvature, then removes the longest edges. These curvature-based methods have been shown to perform competitively against other industry-standard techniques.

A variant of ORC known as the \emph{dynamical ORC} has been used in \cite{gosztolai2021unfolding} to reveal community structures at varying scales of resolution. \emph{Mixed membership}, where a node is allowed to belong to multiple communities has been studied in \cite{tian2023curvature} and \cite{tian2023mixed} (Note that these papers use flow-based approaches similar to \cite{ni2019community}.)  Other curvature notions that are not based on optimal transport have also been used to attack the problem of community detection. One of the popular definitions is that of Forman \cite{forman2003bochner}. Though not equivalent to Ollivier's notion,it has the advantage of being computationally efficient, and its flow version has been used in community detection in \cite{weber2017characterizing}. The shortcomings of Forman's definition have been partially mitigated by modifying its definition (while retaining computational efficiency): one such is the Augmented Forman-Ricci Curvature, which has also been shown to be effective for community detection in \cite{fesser2023augmentations}.

From a theoretical point of view, ORC has been linked to other deep properties of graphs. Besides the work of Ollivier himself, one of the earliest theoretical analyses of ORC \cite{lin2011ricci}, which proves several geometrical and spectral properties of the ORC. Its relation to the heat flow on graphs has been studied in \cite{munch2019ollivier}. The spectrum of the Laplacian is also intimately connected to ORC, and this is shown in \cite{bauer2013ollivierricci}. Analyses regarding geometric properties of ORC such as flatness, rigidity have been studied in \cite{cushing2020rigidity}, \cite{cushing2021flatness}. In spite of the great theoretical interest in ORC and its relation to other properties of graphs, and a great practical interest in its application to community detection, there has not been a substantial amount of literature in the theoretical relation between ORC and community structure.

A basic question regarding this relationship is the following: it is observed that a single edge connecting two disjoint communities will have negative curvature, whereas a complete graph formed by all possible intercommunity edges will have positively curved edges. (This is discussed in Section \ref{sec:prelims}.) Therefore, the critical question is: what is the maximum number of intercommunity edges such that we can guarantee each one is negatively curved? More broadly, when examining two particular communities in a graph, what can we infer about the curvature of intercommunity edges? This paper aims to address this question (and more) quantitatively, as provided by the main theorem.

\begin{theorem} \label{thm:main_thm}
    Suppose $\G$ is a graph comprised of several communities. Let $C_i,C_j$ be distinguished communities in $\G$ whose sizes are $m$ and $n$. Let $k$ be the total number of edges that are either intercommunity edges between $C_i$ and $C_j$, or from any other community to $C_i$ or $C_j$.

    Then, if 
    \[k \le \frac{-m+\sqrt{m^2+4(m-1)(2n-1)}}{2},\]
    we have $\kappa(e) \leq 0$ for every intercommunity edge $e$ between $C_i$ and $C_j$.

    In particular, if $k \le \min_l|C_l|-1$ where $C_l$ are the various communities, the same conclusion holds.
\end{theorem}

We have the following when the community sizes are all the same:

\begin{corollary}\label{cor:1}
    Let $\G$ be a graph as in previous theorem. Suppose in addition that all communities have the same size $n$. Then if $k \le n-1$, we have $\kappa(e) \leq 0$ for every intercommunity edge $e$ between $C_1$ and $C_2$.
\end{corollary}

In the case where the graph only has two communities, we have the following:

\begin{corollary}\label{cor:2}
    Suppose $\G$ is a graph comprised of only two communities $C_1,C_2$ whose sizes are $m$ and $n$. Let $k$ be the total number of  intercommunity edges between $C_1$ and $C_2$.

    Then, if 
    \[k \le \frac{-m+\sqrt{m^2+4(m-1)(2n-1)}}{2},\]
    we have $\kappa(e) \leq 0$ for every intercommunity edge $e$ between $C_1$ and $C_2$.

    If $m=n$, the same conclusion holds for $k \le n-1$.
\end{corollary}







The paper is organized as follows. In Section \ref{sec:prelims}, we do a quick review of optimal transport as it pertains to graphs, define the Ollivier-Ricci curvature of an edge of a graph, and give some examples. In Section \ref{sec:zero_curv}, we give a key example that simultaneously motivates the claims in the paper and highlights the computational techniques used to prove the main theorem. In Section \ref{sec:main_thm}, we provide a proof of the main theorem and the corollaries. In Section \ref{sec:empirical}, we show that the bound prescribed in the theorem is sharp, but at the same time we show experimental results that show that we have a lot of latitude even if we go beyond the theoretical limit prescribed in the theorem.

\section{Preliminaries}\label{sec:prelims}

\subsection{Definitions} \label{subsec:def}
Let $\G = (\V,\E)$ be a graph.
\begin{enumerate}
    \item A \emph{community} of $\G$ is a maximally connected subgraph of $\G$.
    \item An edge whose endpoints lie in different communities is called an \emph{intercommmunity edge}.
    \item An edge whose endpoints both lie inside the same community is called an \emph{intracommunity edge}.
\end{enumerate}

We remark that this is definition of community is only one of several possible definitions, but this simplifies our mathematical analysis.

\subsection{Optimal Transportation on Graphs}

 Let $\mu, \nu$ be probability measures on a graph, and let the vertices be enumerated $1,\ldots,n$. A \emph{transference plan} $\pi = (\pi_{ij})$ is an $n \times n$  matrix with nonnegative entries such that $ \sum_j \pi_{ij} = \mu(i)$, and $\sum_i \pi_{ij} = \nu(j)$. More concisely, $\pi$ is a joint probability distribution on $\V \times \V$ whose marginals are $\mu$ and $\nu$. We define $\Pi$ to be the set of all transference plans between $\mu$ and $\nu$.

Closeness between two probability distributions can be measured via the \emph{1-Wasserstein distance} (or \emph{earthmover's distance}) which is defined as follows:

\begin{equation}\label{eqn:was_inf}
     W(\mu, \nu) = W_1(\mu,\nu) \coloneqq \displaystyle \inf_{\pi \in \Pi} \sum_{i,j} \pi_{ij}d(i,j)
\end{equation}

where $d(i,j)$ is the graph distance between vertices $i,j$.

\begin{proposition}
    Under the metric $W$ defined above, the set of all probability measures on $\V$ is a metric space.
\end{proposition}

Note that computing $W$ amounts to solving a linear programming problem. Therefore, by duality, we have an equivalent formulation for $W$ via a maximization problem. To present this formulation, we first define a $1-Lipschitz$ function to be a function $f$ such that $|f(x)-f(y)| \le d(x,y).$ On (combinatorial) graphs, this is equivalent to insisting that the values of $f$ on adjacent nodes do not differ by more than $1.$

\begin{proposition}
    Let $\mathcal{F}$ be the set of $1$-Lipschitz functions on $G$. Then,
    \begin{equation}\label{eqn:was_sup}
        W(\mu,\nu) = \displaystyle\sup_{f\in\mathcal{F}} \left\{\sum_{z\in G} f(z)(\mu(z)-\nu(z))\right\}
    \end{equation}
\end{proposition}
 A function $f$ that achieves this supremum is called a \emph{Kantorovich potential}.

\subsection{The Ollivier-Ricci Curvature}
Ollivier \cite{ollivier2009ricci} defined Ricci curvatures on general Markov chains on metric spaces, with random walks as the building blocks. The following is a simplified exposition for graphs. A \emph{random walk} is a family of probability distributions $\{m_x\}_{x \in \G}$. For any $\alpha \in [0,1]$, the \emph{$\alpha$-lazy random walk} is one where the probability measures at each node are given by 

\[m_x(z) = m^\alpha_x (z) = 
\begin{cases}
   \alpha, \, z=x\\
   \frac{1-\alpha}{d_x}, \, (zx)\in \E\\
   0, \, \text{otherwise}
\end{cases}\]

Finally, the $\alpha$-\emph{Ollivier-Ricci curvature} of an edge $e=(xy)$ is defined as 
\begin{equation}
    \kappa(e) = \kappa^\alpha(e) \coloneqq 1 - \frac{W(m_x,m_y)}{d(x,y)}.
\end{equation}

We deal with combinatorial graphs in this paper, and therefore, all edges have length $1$. Thus the formula reduces to 
\begin{equation}
    \kappa(e) = 1 - W(m_x,m_y).
\end{equation}

\begin{example}
The intuition behind the Ollivier-Ricci curvature is captured by the following examples:
\begin{enumerate} [a.]
    \item An edge of the lattice $\mathbf{Z}^n$ has curvature $\kappa = 0$. So does an edge of the cycle $C_n$ for $n \ge 6.$
    \item An edge of the complete graph $K_n$ has curvature $\kappa=\frac{n(1-\alpha)}{n-1} >0$.
    \item Imagine a `dumbbell' graph comprised of two communities of size $m,n$, joined by a single intercommunity edge. This edge has curvature $\kappa=2(1-\alpha)\left(\frac{1}{m}+\frac{1}{n}-1\right)$, which is negative for all $m,n \ge 3$.
\end{enumerate}
\end{example}

\begin{remark}
    The claims in the above example can be proved by following the general strategy described here.
    \begin{enumerate}[1.]
        \item To establish an upper bound on the curvature, we need a lower bound on $W(m_x,m_y)$. This uses \eqref{eqn:was_sup} and requires prescribing an explicit potential function.
        
        \item To establish a lower bound on the curvature, we need an upper bound on $W(m_x,m_y)$. This uses \eqref{eqn:was_inf} and requires prescribing an explicit transference plan.

        \item If we can show that the curvature is bounded above and below by the same constant $C$, then the curvature equals $C$.
    \end{enumerate}

\end{remark}

We will apply this strategy to an important nontrivial example in the next section.

\section{A Zero-Curvature Example}\label{sec:zero_curv}

To see what motivates the claims of this paper, we start with the case that the two communities have the same size $n \ge 3$. Consider the number of intercommunity edges below which it is guaranteed that every intercommunity edge is nonpositively curved. We show an explicit construction which suggests that this number can be no more than $n-1$. (A proof of the sufficiency of this claim is given by the main theorem of this paper. Sharpness is given by Proposition \ref{prop:pos_curv}.)

\begin{theorem}
    For communities of size $(n,n)$, there is a configuration of $n-1$ edges between the communities such that each intercommunity edge has zero curvature.
\end{theorem}

\begin{proof}
 We give an explicit construction. Let $A,B$ be a complete graph on the vertices $\{a_0,\ldots,a_{n-1}\},\{b_0,\ldots,b_{n-1}\}$ respectively. Add inter-community edges $a_ib_i,i=0,\ldots,n-2$ (See Figure \ref{fig:zero_curv} for an illustration.) We claim that every edge $a_i b_i$ has zero curvature. Indeed by symmetry, we only need to show it for one of them, say $a_0 b_0$.

\begin{figure}
    \centering
    \includegraphics[scale=0.5]{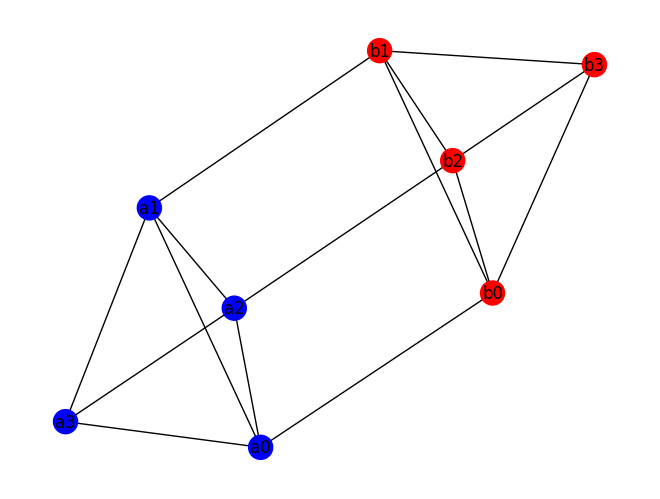}
    \caption{A configuration with zero curvature on all intercommunity edges}
    \label{fig:zero_curv}
\end{figure}

\begin{table}[htbp]
\centering
\begin{tabular}{|c|c|c|c|c|c|c|}
\hline
 - & $a_{n-1}$ & $a_1,\ldots,a_{n-2}$ & $a_0$ & $b_0$ & $b_1,\ldots,b_{n-2}$ & $b_{n-1}$ \\
\hline
$m_{a_0}$ & $\frac{1-\alpha}{n}$ & $\frac{1-\alpha}{n}$ & $\alpha$ & $\frac{1-\alpha}{n}$ & $0$ & $0$\\
\hline
$m_{b_0}$ & 0 & 0 & $\frac{1-\alpha}{n}$ & $\alpha$ & $\frac{1-\alpha}{n}$ & $\frac{1-\alpha}{n}$ \\
\hline
$f$ & 0 & -1 & -1 & -2 & -2 & -3 \\
\hline
\end{tabular}
\caption{A potential function for curvature lower upper bounds}
\label{tab:zero_curv}
\end{table}

To get an upper bound on the curvature, we obtain a lower bound on $W(m_{a_0},m_{b_0})$. We use $f$ as defined in the table \ref{tab:zero_curv}. Applying the equation \eqref{eqn:was_sup}, an explicit calculation shows that $W(m_{a_0},m_{b_0}) \ge 1$, which implies $\kappa(a_0 b_0) \le 0$.

To get a lower bound on curvature, we obtain an upper bound on $W(m_{a_0},m_{b_0})$. We use the following transference plan:
\[\pi_{a_i b_i} = \begin{cases}
                        \alpha-\frac{1-\alpha}{n}, \quad i=0\\
                        \frac{1-\alpha}{n}, \quad i=1,\ldots,n-1
                    \end{cases}\]
together with $\pi_{b_0 b_0} = \frac{1-\alpha}{n}$ and $0$ between any other pair of vertices not specified previously.
Note that the associated distances are
\[d(a_i, b_i) = \begin{cases}
                        1, \quad i=0,\ldots,n-2\\
                        3, \quad i=n-1
                    \end{cases}\]
Applying equation \eqref{eqn:was_inf}, it follows that $W(m_{a_0},m_{b_0}) \le 1$ and hence $\kappa(a_0 b_0) \ge 0$. Thus we conclude that \[\kappa(a_0 b_0) = 0.\]

\end{proof}

\begin{remark}
    The configuration in the previous example is not unique. Indeed, it can be shown using the same technique that the intercommunity edges in Figure \ref{fig:zero_curv_second} have curvature 0.
\end{remark}

\begin{figure}
    \centering
    \includegraphics[scale=0.5]{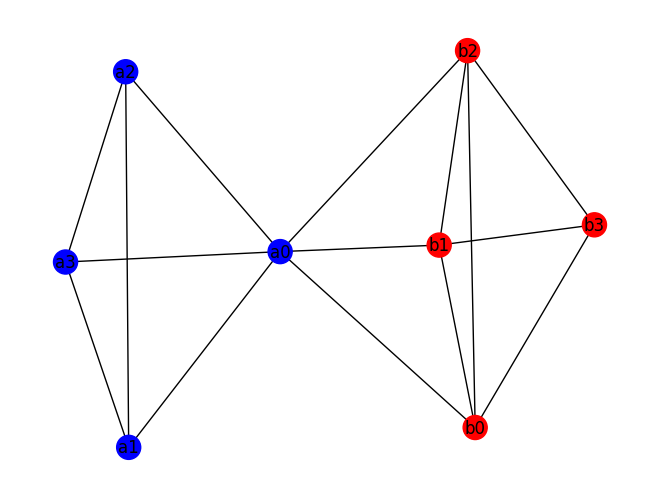}
    \caption{Another configuration with zero curvature on intercommunity edges}
    \label{fig:zero_curv_second}
\end{figure}

\section{Proof of Main Theorem} \label{sec:main_thm}

\begin{figure}[hbt!]
    \centering
    \includegraphics[scale=0.7]{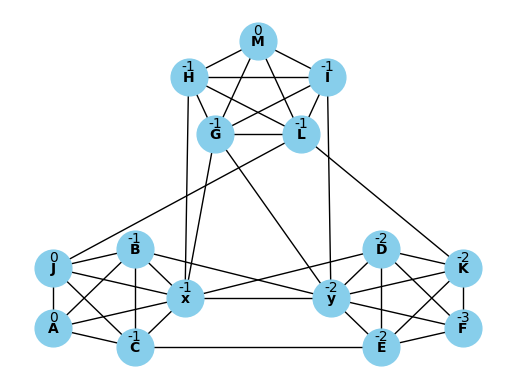}
    \caption{Three communities, with a potential function}
    \label{fig:3comm}
\end{figure}
\begin{center}
    
\end{center}

Let $\G$ be  graph, and $C_1,C_2$ be two communities in $\G$. Suppose  $e=xy \in \E$, $x \in C_1, y \in C_2$. We obtain an upper bound on $\kappa(e)$ by getting a lower bound on $W(m_x,m_y)$. To do so, we need a potential function. The key step is to partition the graph into several subsets on which $f$ will be constant.

Let $C_3 = G - C_1 -C_2$, which could be a union of several communities. We partition $C_1$ in five subsets $\{x\},A,B,C,J$, $C_2$ in $\{y\},D,E,F,K$ and $C_3$ in $G,H,I,L,M$, which we define in the following way:
\begin{itemize}
    \item $B=\{z \in C_1-x : yz\in \E\}$
    \item $D=\{w \in C_2-y : xw\in \E\}$
    
    \item $C=\{z \in C_1-x-B : zw\in \E \,\text{ for some } w\in C_2\}$
    \item $E=\{z \in C_2-y-D : zw\in \E \,\text{ for some } w\in C_2\}$

    \item $J=\{z \in C_1-x-B-C : d(z,w)=2,\text{ for some } w\in C_2\}$
    \item $K=\{w \in C_w-y-B-E : d(z,w)=2,\text{ for some } z\in C_1\}$

    \item $A=C_1-x-B-C-J$
    \item $F=C_2-y-D-E-K$

    \item $G=\{v \in \G-C_1-C_2:\, xv,yv\in\E\}$

    \item $H=\{v \in \G-C_1-C_2-G:\, xv\in\E\}$
    \item $I=\{v \in \G-C_1-C_2-G:\, yv\in\E\}$

    \item $L=\{v \in \G-C_1-C_2-G-H-I:\, zv\in\E \text{ for some } z\in J\}=\{v \in \G-C_1-C_2:\, wv\in\E \text{ for some } w\in K\}$
    \item $M=\G-C_1-C_2-G-H-I-L$

\end{itemize}

This configuration, along with the values of the function $f$ are illustrated in Figure \ref{fig:3comm}. The number shown inside each ``node" (or subset of $\G$ to be precise) is the value of the function $f$. Note that we have not added all possible edges: for instance, there could be edges between $A,M$. The important thing is that $A,J$ do not have neighbours in $C_2$, $F,K$ do not have neighbours in $C_1$. This ensures that the function in Table \ref{tab:dashes} is 1-Lipschitz.

In Table \ref{tab:dashes} we depict the measures $m_x,m_y$ together with the value of $f$. We let $\beta \coloneqq 1-\alpha$.

\begin{table}[h]
\centering
\begin{tabular}{|*{16}{c|}}
\hline
- & $J$ & $A$ & $B$ & $C$ & $x$ & $y$ & $D$ & $E$ & $F$ & $K$ & $G$ & $H$ & $M$ & $I$ & $L$ \\
\hline
$m_x$ & $\frac{\beta}{d_x}$ & $\frac{\beta}{d_x}$ & $\frac{\beta}{d_x}$ & $\frac{\beta}{d_x}$ & $\alpha$ & $\frac{\beta}{d_x}$ & $\frac{\beta}{d_x}$ & 0 & 0 & 0 & $\frac{\beta}{d_x}$ & $\frac{\beta}{d_x}$ & 0 & 0 & 0 \\
\hline
$m_y$ & 0 & 0 & $\frac{\beta}{d_y}$ & 0 & $\frac{\beta}{d_y}$ & $\alpha$ & $\frac{\beta}{d_y}$ & $\frac{\beta}{d_y}$ & $\frac{\beta}{d_y}$ & $\frac{\beta}{d_y}$ & $\frac{\beta}{d_y}$ & 0 & 0 & $\frac{\beta}{d_y}$ & 0 \\
\hline
$f$ & 0 & 0 & -1 & -1 & -1 & -2 & -2 & -2 & -3 & -2 & -1 & -1 & 0 & -1 & -1 \\
\hline
\end{tabular}
\caption{Probabilities and potential associated with curvature of edge $xy$ }
\label{tab:dashes}
\end{table}

In the following we abuse notation by conflating $S$ with its cardinality $|S|$ where $S$ could be any of the sets $A,\ldots,M$. Let $W=W(m_x,m_y)$. Using \eqref{eqn:was_sup} we have the following bound:
\begin{equation}\label{eqn:dist_3comm}
\begin{split}
        W &\ge \alpha + (1-\alpha) \left[ \frac{-1-C-2-2D-G}{d_x} +\frac{-1+B+2D+2E+2K+3F+G+I}{d_y} \right]
\end{split}
\end{equation}

We also have the following constraint equations from counting nodes and edges:

\begin{equation}
\label{eqn:constraint_3comm}
\begin{split}
    A+B+C+J+1 &= n,\\
    D+E+F+K+1 &= m,\\
    n+D+G+H &= d_x,\\
    m+B+G+I &= d_y.
\end{split}
\end{equation}

Equations \eqref{eqn:dist_3comm} and \eqref{eqn:constraint_3comm} together give
\begin{equation} \label{eqn:wass_estim}
    W \ge \alpha + (1-\alpha) \left[\frac{n-1+A+J+G+H}{d_x}+\frac{m-1+F}{d_y}-1\right].
\end{equation}
Note that the lack of symmetry between $x$ and $y$ in the above expression is due to the lack of symmetry in the way we defined $f$.

Let $k_1$ be the number of intercommunity edges between $C_1$ and $C_2$, and $k_2$ be the number of edges from $C_1$ or $C_2$ to the rest of $\G$. Then,

\begin{equation}
    \begin{split}
        k_1 &\ge 1+B+D+\max\{C,E\},\\
        k_2 &\ge 2G+H+I+\max\{J,L\}+\max\{K,L\}.
    \end{split}
\end{equation}

Defining $k=k_1+k_2$, we have
\begin{equation}
    k \geq 1+B+D+\max\{C,E\}+2G+H+I+\max\{J,L\}+\max\{K,L\}
\end{equation}
which, together with the constraint equations \eqref{eqn:constraint_3comm}, yields
\begin{equation}
    \begin{split}
        k+A &\ge n+D+2G+H+I+\max\{K,L\}\\
        &= d_x+G+I+\max\{K,L\}\\
        &\ge d_x
    \end{split}
\end{equation}
and
\begin{equation}
    \begin{split}
        k+F &\ge m+B+2G+H+I+\max\{J,L\}\\
        &= d_y+G+I+\max\{J,L\}\\
        &\ge d_y.
    \end{split}
\end{equation}

Now we need to find an optimal $k$ such that

\[ \frac{n-1+A+J+G+H}{d_x}+\frac{m-1+F}{d_y}-1
    \ge  \frac{k+A+J+G+H}{d_x}+\frac{k+F}{d_y}-1,\]
which reduces to 
\begin{align*}
    &\frac{n-1}{d_x}+\frac{m-1}{d_y}
    \ge  \frac{k}{d_x}+\frac{k}{d_y}\\
    \iff &(m-1)d_x+(n-1)d_y \ge k(d_x+d_y)\\
    \iff &(m-n)d_x+(n-1)(d_x+d_y) \ge k(d_x+d_y)\\
    \iff &(m-n)d_x \ge (k-n+1)(d_x+d_y) \\
    \iff & \frac{d_x}{d_x+d_y} \ge \frac{k-n+1}{m-n}
\end{align*}

Since $d_x \ge n$ and $d_x + d_y \le n+m+k-1$, a sufficient condition for the above to be true is 
\[ \frac{n}{n+m+k-1} \ge \frac{k-n+1}{m-n},\]
which results in the quadratic inequality
\[ k^2 + mk - (m-1)(2n-1) \le 0.\]

Theorem \ref{thm:main_thm} now follows directly from this.

\emph{Proof of Corollary 1:}

    By letting $m=n$, we get
    \begin{align*}
       k \le &\frac{-m+\sqrt{m^2+4(m-1)(2n-1)}}{2}\\
        =&\frac{-n+\sqrt{n^2+4(n-1)(2n-1)}}{2}\\
        =&\frac{-n+\sqrt{n^2+4(2n^2-3n+1)}}{2}\\
        =&\frac{-n+\sqrt{9n^2-12n+4}}{2}\\
        =&\frac{-n+3n-2}{2}\\
        =& n-1. \qed
    \end{align*}

\emph{Proof of Corollary 2:}

In the proof of main theorem, we may assume that $G,H,L,I,M = \emptyset$. Consequently, $J,K = \emptyset$ as well, and $k_2=0$, and the claim follows immediately. \qed

\section{Empirical Results} \label{sec:empirical}
In the previous section, we gave theoretical bounds on the number of intercommunity edges that guarantee the negativity of curvatures. This bound is sharp due to the following proposition:

\begin{proposition} \label{prop:pos_curv}
    For communities of size $(n,n)$, there is a configuration of $n$ edges between the communities such that each intercommunity edge has positive curvature. 
\end{proposition}

\begin{proof}
    Consider the product graph of the complete graph $K_n$ with the graph consisting of only two points joined by an edge. See Figure \ref{fig:pos_config} for an illustration with $n=4$.
    \begin{figure}
    \centering
    \includegraphics[scale=0.5]{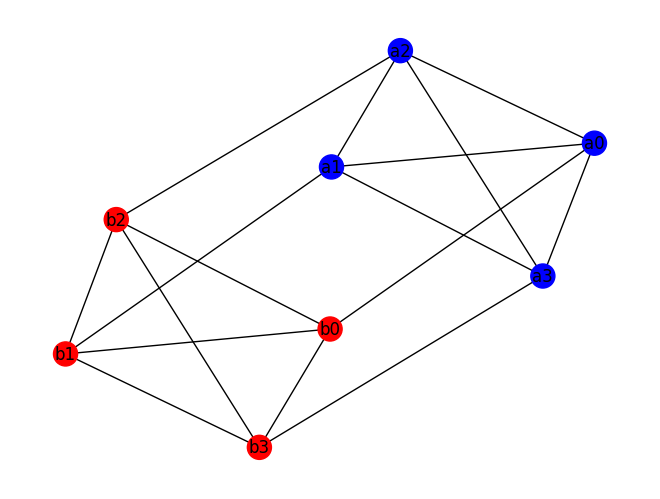}
    \caption{Two communities with all $n$ intercommunity edges positively curved}
    \label{fig:pos_config}
    \end{figure}

    We can show that edge $a_0 b_0$ (and hence every intercommunity edge) has positive curvature. The distributions $m_{a_0},m_{b_0}$ are
    \begin{center}
    \begin{tabular}{|c|c|c|c|c|}
      \hline
      - & $a_0$ & $a_1,\ldots,a_n$ & $b_0$ & $b_1,\ldots,b_n$ \\
      \hline
      $m_{a_0}$ & $\alpha$ & $\frac{1-\alpha}{n}$ & $\frac{1-\alpha}{n}$ & 0 \\
      \hline
      $m_{b_0}$ & $\frac{1-\alpha}{n}$ & $0$ & $\alpha$ & $\frac{1-\alpha}{n}$ \\
      \hline
    \end{tabular}        
    \end{center}

     Consider the transference plan 
    \begin{equation}
        \pi_{a_i b_i}=
        \begin{cases}
            \alpha-\frac{1-\alpha}{n}, \quad i=j=0\\
            \frac{1-\alpha}{n}, \quad i=j \neq 0 \\
            0, \quad i \neq j        
        \end{cases}
    \end{equation}
     This plan has a cost of 
    \[\sum_{i,j} \pi_{a_i b_j}d(a_i,b_j)=\alpha + \frac{(n-2)(1-\alpha)}{n} = 1 - \frac{2(1-\alpha)}{n}\]
    which is an upper bound on $W(m_{a_0},m_{b_0})$, and hence 
    \[ \kappa(a_0 b_0) \ge  \frac{2(1-\alpha)}{n} > 0.\]
\end{proof}

Even though it is theoretically possible for all edges to be positively curved when we have $k=n$ intercommunity edges, the point of the remainder of this section is to share experimental findings that indicate how unlikely such a situation is. For the sake of simplicity, we generate two-community graphs with randomly chosen intercommunity edges and observe the empirical proportion of nonpositively curved edges.


Figure \ref{fig:prop_nonpos_128} shows the result of one such experiment. Here, we choose the community sizes $|C_1|=|C_2|=128$. We experiment with $k=128, 256, 384, 512$ intercommunity edges. For each $k$, we generate 100 graphs where $k$ intercommunity edges are chosen at random. In each of those random graphs, we compute the proportion $p^k_{\le 0}$ of nonpositively curved edges.  Finally, we plot $p^k_{\le 0}$ versus its frequency of occurrence.

One notices in Figure \ref{fig:prop_nonpos_128} that for $k=128, 256$, almost all edges were negatively curved in every one of the 100 randomly generated graphs. When $k=384$, most of the edges are negatively curved. But when $k=512$, the proportion of negatively curved edges is small.

\begin{figure}
    \centering
    \includegraphics[scale=0.5]{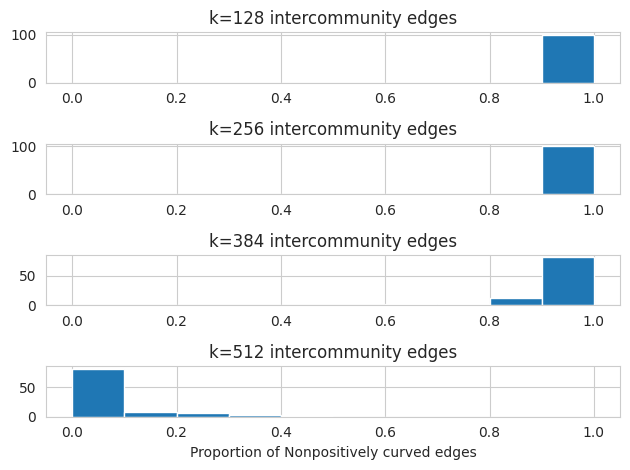}
    \caption{Distribution of Proportion of Nonpositively Curved Intercommunity Edges}
    \label{fig:prop_nonpos_128}
\end{figure}

In fact, we find something similar for different sizes. In Figure \ref{fig:nonpos_k_n}, we examine community sizes $n=16, 32, 64, 128, 256, 512$. And for each $n$, we generate graphs with number of intercommunity edges $k= n, 2n, 3n, 4n$ and empirically find the proportion of intercommunity edges that have nonpositive curvature. For each choice of $(n,k)$, we generate $100$ graphs at random. We plot the empirical proportion of nonpositive edges with error bars one standard deviation wide (over the 100 runs.) What we find is that for $n \ge 32$, even when we have $2n$ intercommunity edges, almost all of them are negatively curved. When $n=512$, even when $k=4n$ we have almost all edges negatively curved.

\begin{figure}
    \centering
    \includegraphics[scale=0.5]{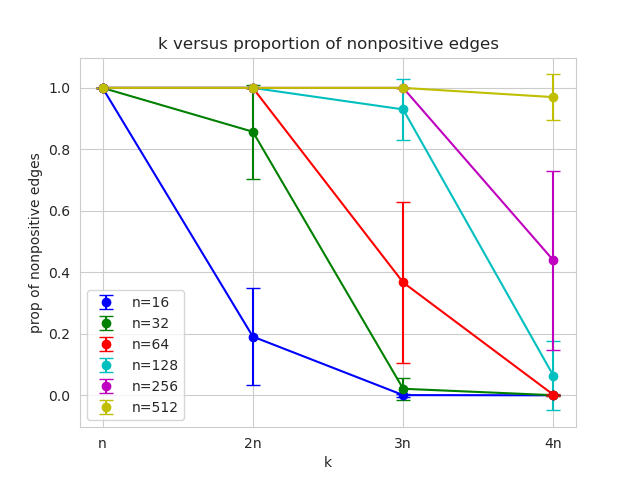}
    \caption{Proportion of Negative Edges for different $k$ and $n$}
    \label{fig:nonpos_k_n}
\end{figure}

\section{Conclusion}

In this paper, we have examined the relation between curvature and community structure from a theoretical point of view. In particular, we have sought to understand how the sizes of communities and the number of intercommunity edges affects the curvature of those intercommunity edges. More specifically, we have given sufficient conditions for intercommunity edges to be negatively curved. In addition, we show these requirements are sharp, in the sense that there are counterexamples as soon as we cross the cutoff. We have achieved this by exploiting two equivalent definitions of the Wasserstein distance between probability distributions, which give us concrete computational tools for proving curvature estimates.

In the experimental section of the paper, we explored how likely it is to find positively curved edges when the number of intercommunity edges $k$ exceeds the theoretical bound. We found that as the community sizes become large, $k$ can get much larger than the theoretical bound while most of the intercommunity edges have negative curvature. A likely explanation for this is in the estimate \eqref{eqn:wass_estim}. Here we see that the Wasserstein distance is large when the ``unmatched" sections $A,F$ are large and the node degrees $d_x,d_y$ are small. When we randomly sample intercommunity edges from the list of all possible intercommunity edges, we are less likely to sample an edge configuration where degrees are very large and the unmatched sections are very small, at least when the number of intercommunity edges is very small.

We believe that this analysis raises several interesting questions. For instance, it would be very interesting to study the curvature of intracommunity edges and obtain criteria that ensure that they are positively curved. In effect, this would provide a theoretical underpinning for curvature-based community detection via edge deletion. Another interesting direction is the analysis of curvature distribution as a function of the number of intercommunity edges. For instance, for a fixed $k$, let $I$ be a random sample of $k$ intercommunity edges from the list of all possible ones. Now we can generate a graph with $I$ as the set of intercommunity edges. Let $\kappa_{max}(I)$ be the maximum curvature among edges in $I$, which we may regard as a random variable. What is the probability distribution of $\kappa_{max}$? How is it affected by $k$? Another interesting question is one that concerns ``surgery", or edge deletion. How does the deletion of a subset of edges of a graph affect the curvature of the remaining edges? Questions of this nature provide ideal circumstances for the synthesis of tools from differential geometry, graph theory, statistics, and programming.

\subsection*{Acknowledgements}
The work of authors SR and VM was partially supported by the US Army Research Office Grant No W911NF2110094.
The work of TB was supported through grant DMS-2105026 of the National Science Foundation.

\bibliographystyle{plain}
\bibliography{references}

\end{document}